\documentclass[12pt]{article}
\usepackage{tikz}
\usepackage{geometry}
\usepackage[english]{babel}
\usepackage[utf8]{inputenc}
\usepackage[T1]{fontenc}
\usepackage{indentfirst}
\usepackage{amsmath}
\usepackage{amsthm}
\usepackage{proof}
\usepackage{amssymb}
\usepackage{eufrak}
\usepackage{graphicx}
\usepackage{psfrag}

\allowhyphens

\newcommand{\complex}{\mathbb{C}}
\newcommand{\egesz}{\mathbb{Z}}

\newcommand{\ordo}{\mathcal{O}}
\newcommand{\lan}{{\boldsymbol \lambda}_N}

\newcommand{\mB}{\mathcal{B}}
\newcommand{\dperdk}{\frac{d}{d\kappa}}
\newcommand{\psivev}[1]{\bra{\Psi}#1\ket{\Psi}}
\newcommand{\lanvev}[1]{\bra{\lan}#1\ket{\lan}}
\newtheorem{thm}{Theorem}

\newcommand{\ket}[1]{{\left|#1\right\rangle}}
\newcommand{\bra}[1]{{\left\langle #1\right|}}

\numberwithin{equation}{section}

\usepackage{geometry}

\usepackage{ifpdf}

\ifpdf
\usepackage{epstopdf}
\usepackage[pdftex,colorlinks,urlcolor=blue,citecolor=blue,linkcolor=blue]{hyperref}
\else
\usepackage[hypertex,colorlinks,urlcolor=blue,citecolor=blue,linkcolor=blue]{hyperref}
\fi
\begin{document}
\title{
\begin{flushright}
\small{CERN-TH-2019-200}
\end{flushright}
\vspace{.5cm}
$T\bar T$-deformation and long range spin chains}
\author{Bal\'azs Pozsgay$^1$, Yunfeng Jiang$^2$, G\'abor Tak\'acs$^3$\\
\parbox{0.7\textwidth}{\center  \small
  $^1$\textit{
MTA-BME Quantum Dynamics and Correlations Research Group,
    Department of Theoretical Physics,\\ Budapest University
    of Technology and Economics,\\ 1521 Budapest, Hungary}\\
  $^2$\textit{Theoretical Physics Department, CERN, \\1211 Geneva 23, Switzerland.}\\
$^3$\textit{
BME ``Momentum'' Statistical Field Theory Research Group,\\
    Department of Theoretical Physics,\\ Budapest University
of Technology and Economics,\\ 1521 Budapest, Hungary}
\small{\date{25th November 2019}}   
}}
\maketitle

\abstract{We point out that two classes of deformations of integrable models,
    developed completely independently, have deep connections and share the same algebraic
    origin. One class includes the $T\bar T$-deformation of  1+1 dimensional  integrable quantum
    field theory and related solvable irrelevant deformations proposed recently. The other class is
    a specific type of long range integrable deformation of quantum spin chains introduced a decade
    ago, in the context of  $\mathcal{N}=4$ super-Yang-Mills theory. We show that the detailed
    structures of the two deformations are formally identical and therefore share many features. Both
    deformations preserve integrability and lead to non-local deformed theories, resulting in a
    change of the corresponding  factorized $S$-matrices.  We also prove a factorisation formula for
    the expectation value of the operators which trigger the deformation on the lattice; similar
    results in quantum field theory play an essential role in the solvability of such
    deformations. We point out that the long range deformation is a natural counterpart of the
    $T\bar{T}$-deformation for integrable spin chains, and argue that this observation leads to
    interesting new avenues to explore.}

\section{Introduction}

Recently, a class of solvable irrelevant deformations of quantum field theories (QFT) have attracted
considerable attention. One of the most well-studied example of such solvable deformations
is the so-called $T\bar{T}$ deformation \cite{Smirnov-Zam-TTbar,Cavaglia:2016oda} (see also
\cite{Zamolodchikov:1991vx,mussardo-ttbar}), which can be 
defined for any 2d QFT. The $T\bar{T}$ deformation has many distinguished features, among which the
following two are most relevant to us:
\begin{enumerate}
\item \textbf{Solvability and integrability}. Usually, irrelevant deformations of QFT are highly
  ambiguous and complicated due to appearance of an infinite number of counter-terms. In contrast, the $T\bar{T}$ deformation
  is under much better control and solvable in an appropriate sense. In particular, when the original theory is an integrable quantum
  field theory (IQFT), the deformation preserves integrability. This allows for exact analytical results,
  especially when the original theory is a conformal field theory (CFT) or IQFT (cf. the reviews \cite{Giveon:2019fgr,Jiang:2019hxb} and references therein).
\item \textbf{Non-locality}. The deformed QFT becomes non-local, which can be seen classically by
  reformulating the $T\bar{T}$ deformation as coupling to various 2d gravity theories including the
  Jackiw-Teitelboim type gravity \cite{Dubovsky:2017cnj,Dubovsky:2018bmo}, massive ghost-free
  gravity \cite{Tolley:2019nmm} and random geometry \cite{Cardy-TTbar1,Cardy-TTbar2}. At the quantum level
  it can be seen from the fact that the asymptotic density of states exhibit a Hagedorn behaviour instead of
  the usual Cardy growth \cite{HagedornTTbar,Datta:2018thy,Aharony:2018bad}.
  Therefore the non-locality of $T\bar{T}$ deformed theories makes them different from usual local QFT and leads to
  a novel type of UV behaviour.
\end{enumerate}
Other solvable irrelevant deformations including the $J\bar{T}$ deformation \cite{Guica:2017lia},
deformation by higher spin irrelevant operators constructed from KdV currents
\cite{LeFloch:2019rut,LeFloch:2019wlf} and their various combinations
\cite{Chakraborty:2019mdf,Frolov:2019xzi}, which also share these two features. Another common characteristics of
these deformations is that the irrelevant operators which trigger such deformations are constructed
from certain bi-local combination of conserved currents. Exploiting the conservation of the currents and translational
invariance leads to a factorisation formula for the expectation value of such irrelevant
operators. This simple yet crucial observation was first pointed out by Zamolodchikov in 2004
\cite{Zam-TTbar} for the $T\bar{T}$ operator and lies at the heart of solvability of these
deformations.

Given the recent development of solvable irrelevant deformations for QFT, a natural and intriguing
question is whether such deformations exist for \emph{integrable lattice models} such as quantum
spin chains, as it was stated in the paper \cite{Smirnov-Zam-TTbar}: {\it ``Connections of such
``effective theories'' with lattice integrability seems an interesting question to
explore.''}. In this work we point out that not only do such deformations exist for integrable quantum
spin chains, but they had been constructed
already a decade ago, although in disguise! The $T\bar{T}$-like deformation corresponds to a
specific type of integrable long-range deformation studied in \cite{beisert-long-range-2}, where
the generator of the deformation is a bi-local combination of the charges.

Let us discuss this point in more detail. The spin chain analogue of the  $T\bar{T}$-deformation in integrable QFT
should capture the two important features mentioned above, namely it must \emph{preserve
  integrability} and the deformed theory should \emph{become non-local} in an appropriate sense. Usual
integrable quantum spin chains such as the Heisenberg spin chain involve only nearest neighbor
interactions. Spin chains that involve interactions between more sites are called long-range spin
chains and are much less studied.  Interestingly, a wide class of long range deformed spin chains
with the desired properties were introduced by Bargheer, Beisert and Loebbert in
\cite{beisert-long-range-2} (see also \cite{long-range-boundary,beisert-xxz-deformation}).
The motivation of these studies was to find a systematic way to construct the
dilatation operator in planar $\mathcal{N}=4$ super-Yang-Mills theory at higher loop orders. For
that purpose, the requirements are essentially the same: the dilatation operators at higher orders
should preserve integrability and the interacting range should increase order by order in the 't Hooft
coupling. The emerging non-locality is not very severe: the deformed Hamiltonian and higher
conserved charges remain extensive and quasi-local (see the main text for details).

According to \cite{beisert-long-range-2}, there are two main types of non-trivial long-range deformations which
preserve integrability: the boost-type and the bi-local-type. The boost-type deformation modifies the dispersion relation,
while the bi-local-type deformation modifies the factorised $S$-matrix of the quasi-particle excitations.

The boost-type deformations were treated in detail by one of the authors in \cite{sajat-longcurrents}, where
it was realised that they involve certain generalised current operators. The physical
meaning of these generalised currents is rather simple: they describe the flow of a conserved quantity
under a time evolution generated by some other conserved charge. As special cases they also include
the physical current operators, which describe the flow under the fundamental spin chain Hamiltonian.

Another motivation for the study of such current operators originates from the recently introduced
Generalised Hydrodynamics (GHD), which describes the large scale dynamics of integrable models out
of equilibrium. To the leading order in a hydrodynamic expansion it captures the ballistic part of the
transport \cite{doyon-GHD,jacopo-GHD}. The theory is based on a local density approximation and the
continuity relations for the conserved charges, and it leads to a generalised Boltzmann-type
equation involving all charges. For this purpose it is essential to know the mean values of the
currents in any finite density state. A physically motivated conjecture for the currents was given in
\cite{doyon-GHD,jacopo-GHD}. For massive integrable QFT it was already proven in \cite{doyon-GHD},
and a proof for the spin chains was given in \cite{sajat-currents}. Later it was realised in
\cite{sajat-longcurrents} 
that the connection to the boost-deformed long range spin chains provides a rather simple derivation of the
results of \cite{sajat-currents}.

In the present work we treat the long range deformations of the bi-local type for which the generator of
the deformation is a certain bi-local combination of the conserved charges. As explained below, after proper rewriting
the perturbing operators take essentially the same form as the perturbing irrelevant operators in QFT: they are composed
of an anti-symmetric combination of charge densities and generalised currents, a connection which went unnoticed in
\cite{beisert-long-range-2}. Applying Zamolodchikov's argument to the lattice case we prove a factorisation formula for the perturbing operators on the lattice. This demonstrates that the algebraic construction is completely the same as for the $T\bar{T}$ deformation in the QFT. As a direct consequence, the $S$-matrix is deformed by a phase factor involving the charge eigenvalues, which
echoes the fact that the $T\bar{T}$-deformation modifies the $S$-matrix of QFT by multiplying it with a CDD factor.

The paper is composed as follows. In section~\ref{sec:TTbar}, we give a brief review of the salient
features of the solvable irrelevant deformations. In section~\ref{sec:localchain}, we introduce
local integrable spin chains and a class of operators constructed from specific combinations of the
current and charge density operators. We then prove a lattice version of the factorisation formula
for the mean values of these operators. In section~\ref{sec:longchain}, we discuss integrable
long-range deformations of local spin chains. We mainly focus on the bi-local-type deformation and
point out the intimate relationship between this deformation and the solvable irrelevant
deformations of QFT. In section~\ref{sec:asymp}, we consider the deformed expectation value of
conserved charges in the finite volume using asymptotic Bethe Ansatz. We confirm that the
factorisation formula still holds in the finite volume. Finally we conclude in
section~\ref{sec:discuss} and discuss future interesting directions to explore.

\section{$T\bar T$-deformation of QFT}

\label{sec:TTbar}
In this section we provide a short review of the solvable irrelevant deformations of QFT that are related to the discussion in this work; the reader is invited to consult the original papers for further details.
\paragraph{Definitions} We give the definition of general solvable deformation in Lagrangian
formalism following  {\cite{LeFloch:2019rut,Bonelli-stb-ttbar-classical}}. For a given 2d QFT
described by an action $S_0$, 
consider two conserved currents $J^{(1)}_{\mu}$ and $J^{(2)}_{\mu}$ satisfying $\partial^{\mu}J^{(a)}_{\mu}=0$.
Using these currents we construct the following composite operator
\begin{align}
  \label{QFTeps}
\mathcal{O}\equiv \epsilon^{\mu\nu}J_{\mu}^{(1)}J^{(2)}_{\nu}.
\end{align}
Taking $J^{(1)}_{\mu}=T_{1\mu}$ and $J^{(2)}=T_{2\mu}$, the operator $\mathcal{O}$ is called the
$T\bar{T}$ operator. More precisely, we have
\begin{align}
\mathcal{O}=\epsilon^{\mu\nu}T_{1\mu}T_{2\nu}=T_{11}T_{22}-T_{12}T_{21}=\det T_{\mu\nu}
\end{align}
which can be written in the complex coordinates as \cite{Smirnov-Zam-TTbar}
\begin{align}
\det T_{\mu\nu}=-(T\bar{T}-\Theta^2)
\end{align}
where $T$, $\bar{T}$ and $\Theta$ are defined as
\begin{align}
T=&\,-\frac{1}{2}(T_{11}-T_{22}-2i T_{12}),\\\nonumber
\bar{T}=&\,-\frac{1}{2}(T_{11}-T_{22}+2i T_{12}),\\\nonumber
\Theta=&\,\frac{1}{2}(T_{11}+T_{22}).
\end{align}
For CFT, $\Theta=0$ and the operator takes the form of $T\bar{T}$.

Using the composite operator $\mathcal{O}$ we can define a family of theories parametrized by a
parameter $\lambda$ 
\begin{align}
\label{eq:defS}
\frac{d S_{\lambda}}{d\lambda}=\int d^2x\,\mathcal{O}_{\lambda}(x).
\end{align}
We stress that under this deformation the currents $J_{\mu}^{(a)}$ remain conserved, but their
explicit form changes depending on $\lambda$. Therefore the corresponding operator
$\mathcal{O}_{\lambda}$ is also deformed and depends on $\lambda$. 

\paragraph{Factorization formula} The solvability of the quantum theory for the class of theory (\ref{eq:defS}) is based on the factorization formula of the expectation value of the composite operator. To derive it, consider the theory on a cylinder where the spacial direction is a circle of length $L$, while the temporal direction is non-compact. For a generic eigenstate of the Hamiltonian denoted by $|\psi\rangle$, translational invariance and conservation of the current implies
\begin{align}
  \label{eq:QFTfact}
\langle\psi|\mathcal{O}_{\lambda}|\psi\rangle=\epsilon^{\mu\nu}\langle\psi|J_{\mu}^{(1)}J_{\nu}^{(2)}|\psi\rangle
=\epsilon^{\mu\nu}\langle\psi|J_{\mu}^{(1)}|\psi\rangle\langle\psi|J_{\nu}^{(2)}|\psi\rangle.
\end{align}
As a result, the expectation value of the composite operator $\mathcal{O}_{\lambda}$ can be written
in terms of the expectation values of the currents along the whole flow. The factorization formula
was first derived by Zamolodchikov for $\mathcal{O}_{\lambda}$ being the $T\bar{T}$ operator \cite{Zam-TTbar}.

\paragraph{Deformed spectrum} The factorization formula leads to a flow equation for the deformed energy. Denoting the deformed energy by $\langle\psi|H_{\lambda}|\psi\rangle=E(\lambda,L)$, and using the definition (\ref{eq:defS}) and the Hellmann-Feynman theorem yields
\begin{align}
\label{eq:defE}
\frac{\partial}{\partial\lambda}E=L\,\langle\psi|\mathcal{O}_{\lambda}|\psi\rangle
=L\,\epsilon^{\mu\nu}\langle\psi|J_{\mu}^{(1)}|\psi\rangle\langle\psi|J_{\nu}^{(2)}|\psi\rangle.
\end{align}
Given some convenient expressions for the expectation values of the conserved currents, the above equation can be used to determine the deformed spectrum. For example, taking $\mathcal{O}_{\lambda}$ to be the $T\bar{T}$ operator, the right hand side of (\ref{eq:defE}) can be written in terms of the energy and momentum, so the flow equation reads
\begin{align}
\label{eq:burger}
\frac{\partial E}{\partial\lambda}=\frac{1}{2}\left(E\partial_L E+\frac{P^2}{L}\right),
\end{align}
where $P$ is the momentum of the state $|\psi\rangle$. The equation (\ref{eq:burger}) is the inviscid Burger's equation in one dimension which can be solved by method of characteristics. When the original theory is a CFT, the deformed energy $E(\lambda,L)$ can be found explicitly.

\paragraph{Deformed $S$-matrix and CDD factor} The $S$-matrix of the QFT is deformed in a simple way
under the solvable irrelevant deformation \cite{Smirnov-Zam-TTbar} {(see also
  \cite{mussardo-ttbar})}. This fact is particularly powerful for IQFT since the factorized
$S$-matrix plays an essential role in computing many physical quantities. Let us denote the
factorized $S$-matrix by $S_{ij}^{kl}(\theta)$ where $\theta\equiv\theta_i-\theta_j$ and $\theta_i$
are the rapidities of the particles. Under $T\bar{T}$ deformation the $S$-matrix is deformed as 
\begin{align}
\label{eq:defSQFT}
S_{ij}^{kl}(\theta)\mapsto S_{ij}^{kl}(\theta) e^{i\delta_{ij}^{(\lambda)}(\theta_i,\theta_j)},
\end{align}
where
\begin{align}
\delta_{ij}^{(\lambda)}(\theta_i,\theta_j)=\lambda\,\epsilon_{\mu\nu}p_i^{\mu}p_j^{\nu}
\end{align}
and  $p_i^{\mu}=(E_i,P_i)$ is the 2-momentum of the particle. For a massive relativistic QFT
\begin{align}
p_i^{\mu}=(m_i\cosh\theta_i,m_i\sinh\theta_i)
\end{align}
and the corresponding phase factor takes the form
\begin{align}
\delta_{ij}^{(\lambda)}(\theta)=\lambda m_i m_j\sinh(\theta).
\end{align}
For massless particles
\begin{align}
\delta_{ij}^{(\lambda)}(\theta)=\frac{\lambda}{2}M_iM_je^{\theta}=-2\lambda\, p_i^{(+)}p_j^{(-)},
\end{align}
where $p_i^{(+)}$ and $p_j^{(-)}$ are the momenta of left-moving and right-moving massless particles. Analogous phase factors for the irrelevant deformation triggered by bi-locals of higher KdV charges are given in \cite{Smirnov-Zam-TTbar}, while the phase factor corresponding to $J T_a$ deformation was found recently in \cite{Anous:2019osb}.

In the following we demonstrate that all these features can be generalised in a natural way to long
range deformations of spin chains of bi-local type. 

\section{Local spin chains}
\label{sec:localchain}
We consider integrable local spin chains given by Hamiltonian
\begin{equation}
  H=\sum_{x=1}^L h(x),
\end{equation}
where $h(x)$ is a local operator acting on the nearest neighbor sites $x$ and $x+1$.

As a concrete example we consider integrable spin chains with $SU(N)$ symmetry, where at each site the local
space is $\complex^N$ and
\begin{equation}
  h(x)=P_{x,x+1}-1,
\end{equation}
where $P$ is the permutation operator exchanging the local spaces on sites $x$ and $x+1$; for $N=2$
this corresponds to the XXX Heisenberg spin chain. However, we stress that the discussions below are
quite general and work for other important 
classes of integrable spin chains such as the XXZ and XYZ chains.

Integrable spin chains have a family of infinitely many conserved charges $Q_\alpha$ in
involution. They are extensive operators and can be written as
\begin{equation}
  Q_\alpha=\sum_x q_\alpha(x),\quad \alpha=1,\dots,\infty,
\end{equation}
where $q_\alpha(x)$ are the charge densities. In the following we use the notation $|\ordo(x)|$ for
the range of the local operator 
$\ordo(x)$. The charge densities can be chosen such that $|q_\alpha(x)|=\alpha$, and specifically we
have $H\sim Q_2$, where the proportionality factor depends on conventions.

{The local charges are usually obtained from a commuting set of transfer matrices, which are built from local Lax
operators \cite{Korepin-Book}. Alternatively, they can also be obtained
with help of the boost operator \cite{Tetelman,sogo-wadati-boost-xxz,Thacker-boost,GM-higher-conserved-XXZ},
which is the approach we use below.

For an extensive local operator $L=\sum_{x=-\infty}^\infty l(x)$ we define the corresponding boost
operator as the formal sum 
\begin{equation}
  \mB[L]=\sum_{x=-\infty}^\infty x l(x).
\end{equation}
Then we have  \cite{Tetelman,sogo-wadati-boost-xxz,Thacker-boost,GM-higher-conserved-XXZ}
\begin{equation}
  \label{boostrel}
  Q_{\alpha+1}=i[\mB[Q_2],Q_\alpha]+\text{constant.}
\end{equation}
The constant part can be chosen in various ways; one possibility is to require that the charges have
zero eigenvalues on a specifically chosen ferromagnetic reference state.

The current operators associated to the charges are defined through the continuity equation
\begin{equation}
  \label{qqx}
 \frac{d}{dt} q_\alpha(x)=
  i  \left[H,q_\alpha(x)\right]=J_\alpha(x)-J_\alpha(x+1),
\end{equation}
where the choice for the shift in $x$ is just a matter of convention.

The essential difference between the lattice and the QFT is that space derivatives are discrete and
there is no Lorentz invariance here; in particular, the pair $(q_\alpha(x),J_\alpha(x))$
does not form a vector. Nevertheless, as explained below they play a very similar role in the
deformations as their QFT counterparts $J_\mu$, which are Lorentzian 2-vectors.

We also introduce the generalised currents, that describe the flow of charge $Q_\alpha$ under the time
evolution dictated by $Q_\beta$:
\begin{equation}
  \label{Jab}
i  \left[Q_\beta,  q_\alpha(x)\right]=J_{\alpha,\beta}(x)-J_{\alpha,\beta}(x+1).
\end{equation}
These operators also  play an important role in the long range deformations. The analogous
construction in QFT would correspond to using the higher conserved charges to generate Hamiltonian
time evolution of the system. Such quantum integrable hierarchies exist e.g. for the quantum KdV
system \cite{QuantumKdV} and the quantum Benjamin-Ono equation \cite{QuantumBenjaminOno}.

\subsection{Factorization on the lattice}

Now we construct a certain combination of the current and charge density operators, such that the
mean value of the resulting local operator factorises. We adopt the arguments of
\cite{Cardy-TTbar1,Zam-TTbar} to the lattice case. Let us fix the indices $\alpha\ne\beta$ and consider the operator
\begin{equation}
 J_\alpha(x)q_{\beta}(0)-q_\alpha(x)J_{\beta}(1).
\end{equation}
Setting $x=0$ or $x=1$ we can recognise the lattice version of the anti-symmetric combination
\eqref{QFTeps} from QFT. 

Let us also fix an arbitrary eigenstate $\ket{\Psi}$ of the model; the computations can be performed both in
finite or in infinite volume.

\begin{thm}
  \label{Cx}
The function
\begin{equation}
  C(x)=\bra{\Psi}  J_\alpha(x)q_{\beta}(0)-q_\alpha(x)J_{\beta}(1) \ket{\Psi}
\end{equation}
does not depend on $x$.
\end{thm}
\begin{proof}
Applying a lattice derivative and using the translational invariance of the correlator yields
\begin{equation}
  C(x+1)-C(x)
  =\bra{\Psi}  \Big(J_\alpha(x+1)-J_\alpha(x)\Big)q_{\beta}(0)\ket{\Psi}+
\bra{\Psi}q_\alpha(x)\Big(J_{\beta}(1)-J_\beta(0)\Big) \ket{\Psi}.
\end{equation}
From the definition \eqref{qqx} we get
\begin{equation}
  \begin{split}
  C(x+1)-C(x)&=-i\left[
\bra{\Psi} [H,q_\alpha(x)]q_{\beta}(0)\ket{\Psi}+
\bra{\Psi}q_\alpha(x)[H,q_\beta(0)] \ket{\Psi}
\right]=\\
&=-i\bra{\Psi}
[H,q_\alpha(x)q_\beta(0)]
\ket{\Psi}=0.    \\
  \end{split}
\end{equation}
where in the last step we used the fact that $\ket{\Psi}$ is an eigenvector of $H$.
\end{proof}

Denoting the constant value of $C(x)$ by $C$, using translational invariance and the fact
$\ket{\Psi}$ is an eigenvector of the conserved charges we get
\begin{equation}
LC=  \sum_{x=1}^L C(x)=\psivev{J_\alpha(x)}\psivev{Q_\beta}-\psivev{Q_\alpha}\psivev{J_\beta(x)}.
\end{equation}
Here  the current mean values don't depend on $x$;
we just kept the dependence on $x$ to signal that $J_{\alpha}$ and $J_\beta$ are intensive quantities, while the $Q$'s are extensive.

Dividing by $L$ and using the charge densities again we get
\begin{equation}
  \label{eq:latticefact}
\psivev{J_\alpha(x)q_{\beta}(0)-q_\alpha(x)J_{\beta}(1) }=
  \psivev{J_\alpha}\psivev{q_\beta}-\psivev{q_\alpha}\psivev{J_\beta},
\end{equation}
where now we deleted the $x$-dependence on the r.h.s. As in the field theory case
\cite{Zam-TTbar,Cardy-TTbar1}, the 
derivation only depends on the conservation equation \eqref{qqx} and translational invariance.

This argument can be extended to the generalised currents defined in \eqref{Jab}. Introducing the three index local
operators
\begin{equation}
  \label{Kabc}
  K_{\alpha,\beta,\gamma}(x)=J_{\alpha,\gamma}(x)q_{\beta}(x)-q_\alpha(x)J_{\beta,\gamma}(x+1)
\end{equation}
we obtain:
\begin{thm}
The mean values of $K_{\alpha,\beta,\gamma}(x)$ factorize as
\begin{equation}
  \label{Jabfact}
\psivev{ K_{\alpha,\beta,\gamma}(x)}=
  \psivev{J_{\alpha,\gamma}}\psivev{q_\beta}-\psivev{q_\alpha}\psivev{J_{\beta,\gamma}}.
\end{equation}
\end{thm}
Note that \eqref{eq:latticefact} and \eqref{Jabfact} are precisely the lattice
counterparts of the factorization formula in QFT \eqref{eq:QFTfact}. To our knowledge this result is new. We give a few more comments about the factorisation in Section \ref{sec:discuss}.

Now we show that this particular combination of operators arises from a simple commutation relation,
which is analogous to the boost relation \eqref{boostrel}. Let us pick two indices $\alpha\ne \beta$
and define the formal sum 
\begin{equation}
  X=\sum_{x<y} q_\alpha(x)q_\beta(y).
\end{equation}

\begin{thm}
  \label{XQgthm}
The following commutator generates the factorising local operators:
  \begin{equation}
    \label{XQg}
     i[X,Q_\gamma] =
    \sum_x K_{\alpha,\beta,\gamma}(x).
     \end{equation}
 \end{thm}
\begin{proof}
The commutator can be computed as
\begin{equation}
  \sum_{x<y} \left\{ iq_\alpha(x) [q_\beta(y),Q_\gamma]+
    i[q_\alpha(x),Q_\gamma] q_\beta(y)
\right\}
\end{equation}
which can be rearranged as
\begin{equation}
  \label{opl}
\sum_x  iq_\alpha(x) \left[\sum_{y>x} q_\beta(y),Q_\gamma\right]+
\sum_y  i\left[\sum_{x<y} q_\alpha(x) ,Q_\gamma\right] q_\beta(y).
\end{equation}
Integrating \eqref{Jab} yields
\begin{equation}
  \begin{split}
  \label{Jab2}
i  \left[Q_\beta,  \sum_{y>x} q_\alpha(y)\right]&=J_{\alpha,\beta}(x+1)\\
i  \left[Q_\beta,  \sum_{y<x} q_\alpha(y)\right]&=-J_{\alpha,\beta}(x).
\end{split}
\end{equation}
and substituting these into \eqref{opl} after the appropriate changes yields
\begin{equation}
 -  \sum_x  q_\alpha(x)J_{\beta,\gamma}(x+1)
  +
  \sum_y  J_{\alpha,\gamma}(y) q_\beta(y),
\end{equation}
which proves the theorem.
\end{proof}

Just as in the QFT case, the structural form of $K_{\alpha,\beta,\gamma}$ is given by a determinant
\begin{equation}
  K_{\alpha,\beta,\gamma}(x)=
  \begin{vmatrix}
    J_{\alpha,\gamma}(x) &  q_\alpha(x) \\
J_{\beta,\gamma}(x+1) &    q_\beta(x)
  \end{vmatrix}
\end{equation}
with a special choice for the operator ordering. The case of the stress-energy tensor
(and thus $\det T$) encountered in QFT is obtained by formally setting $Q_\alpha=Q_\gamma=H$ and
$Q_\beta=P$, where $P$ is the total momentum. However, it is not possible to generate the precise
analogue of $T$ on the lattice, because there is no local operator corresponding to the momentum
density. This is also consistent with the fact that Lorentz invariance is lost on the lattice.

On the other hand, one of the advantages of working on the lattice is that there is no need to regularise the product of
operators e.g. by point-splitting. Nevertheless, the prescription is not entirely trivial as it involves a shift in the $x$ coordinate whose precise form
 depends on the definition of the currents given in \eqref{Jab}.

We also note that simple re-definitions of the generating operator $X$ lead essentially to the same result. For example we could have taken
\begin{equation}
  \label{Xl}
 X=\sum_{x<y+l} q_\alpha(x)q_\beta(y)
\end{equation}
with some $l\in \egesz$. This would give
\begin{equation}
  i[X,Q_\gamma]=\sum_x  \left[ J_{\alpha,\gamma}(x+l) q_\beta(x)   -q_\alpha(x+l)J_{\beta,\gamma}(x+1)
\right].
  \end{equation}
Theorem \ref{Cx} states that the mean values of this operator do not depend on $l$, therefore $l$
could be chosen arbitrary. In particular  $l$ can be chosen large enough so that the supports of the
two operators in \eqref{Xl} do not 
overlap in space, for which case there is no issue with the operator ordering.

\section{Long range spin chains}
\label{sec:longchain}
In this section we explain how to deform the local spin chains with the factorizing operators
described above. The framework for this was developed in \cite{beisert-long-range-1,beisert-long-range-2},
although  the relation of  the perturbing operators to the generalised currents was not realised back then.
We restrict ourselves to the infinite volume situation in what follows, except in Section  \ref{sec:asymp} where we return to finite volumes.

We introduce a deformation parameter $\kappa$ and set out to find the deformations the commuting set of charges
which satisfy the following conditions:
\begin{enumerate}
\item The deformed charges allow a power series expansion in $\kappa$
\begin{equation}
  Q_\alpha^\kappa=\sum_{j=0}^\infty \frac{\kappa^j}{j!}  Q_\alpha^{(j)}
\end{equation}
with the initial conditions
\begin{equation}
  Q_\alpha^{\kappa=0}=Q_\alpha,
\end{equation}
where $Q_\alpha$ are the local charges of the original (local) spin chain.
\item The charges continue to form a commuting family
\begin{equation}
  [Q_\alpha^\kappa,Q_\beta^\kappa]=0,
\end{equation}
which ensures that integrability is preserved by the deformation.
\item We also require that the resulting operators remain extensive
  and quasi-local\footnote{
     Following
\cite{prosen-xxx-quasi,prosen-enej-quasi-local-review} we call an
extensive operator $A=\sum_x a(x)$
quasi-local, if the Hilbert-Schmidt (HS) norm of its traceless part grows
at most linearly with the volume, and if its HS overlap with any local
operator is finite in the $L\to\infty$ limit.
way.
Quasi-local operators can include pieces with arbitrary long
range, but the amplitudes of these terms typically decay
exponentially with the range.
}, written as
\begin{equation}
  \label{Qxdef2}
  Q_\alpha^\kappa=\sum_{x=-\infty}^\infty q^\kappa_\alpha(x).
\end{equation}
in terms of appropriate charge densities $q^\kappa_\alpha(x)$.
\item Furthermore the deformation of the infinite volume eigenstates can be expressed as a power series
\begin{equation}
  \ket{\Psi^\kappa}=\sum_{j=0}^\infty \frac{\kappa^j}{j!}  \ket{\Psi^{(j)}},
\end{equation}
where $\ket{\Psi^{(0)}}$ are the original eigenstates.
\end{enumerate}

These requirements can easily be satisfied by postulating the following formal generating equations
\cite{beisert-long-range-1,beisert-long-range-2}:
\begin{equation}
  \begin{split}
  \label{defdef}
  \frac{d}{d\kappa} \ket{\Psi^\kappa}&=-iX(\kappa) \ket{\Psi^\kappa}   \\
  \frac{d}{d\kappa} Q_\alpha^\kappa&= i[X(\kappa),Q_\alpha^\kappa],\\
\end{split}
\end{equation}
where $X(\kappa)$ is a formal operator to be specified later, typically dependent on $\kappa$.

The generating equation naturally preserves the commutation of the charges due to
\begin{equation}
  \frac{d}{d\kappa} [Q_\alpha^\kappa,Q_\beta^\kappa]=
  i[X(\kappa), [Q_\alpha^\kappa,Q_\beta^\kappa]],
\end{equation}
where the Jacobi identity was exploited, and the eigenvalues of the charges are also unchanged:
\begin{equation}
  Q_\alpha^\kappa\ket{\Psi^\kappa}=\Lambda_\alpha^\Psi \ket{\Psi^\kappa},
\end{equation}
where $\Lambda_\alpha^\Psi$ are the eigenvalues corresponding to the state $\ket{\Psi}$ in the original model.

These simple consequences of the generating equation do not depend on the form of $X(\kappa)$, which
is instead constrained by the 
physical requirement that the deformed charges should remain quasi-local.

In \cite{beisert-long-range-2} three families of deformations satisfying this requirement were
identified\footnote{The quasi-locality of the deformed charges was not
proven there, but it is clearly satisfied at least in some neighborhood
of $\kappa=0$.}:
\begin{enumerate}
\item {\bf Local operators.} Take
  \begin{equation}
    X=\sum_{x=-\infty}^\infty \ordo(x)
  \end{equation}
  with $\ordo$ being any short range operator. This deformation describes a ``physical'' similarity
  transformation corresponding to a change of
  basis:
\begin{equation}
  Q_\alpha^\kappa=e^{i\kappa X} Q^0_\alpha e^{-i\kappa X },
\end{equation}
and can be extended immediately to the case when $\ordo(x)$ is quasi-local.

\item {\bf Boost operators.} The choice
  \begin{equation}
    \label{Xboost}
    X(\kappa)=-\mB[Q_\alpha^\kappa]
  \end{equation}
  for some $\alpha$ also generates quasi-local charges. This deformation was treated in detail in
  \cite{sajat-longcurrents}, where it was shown that the $\kappa$-derivative of the charges is
  \begin{equation}
    \label{qabjabk}
    \dperdk Q_\beta^\kappa= \sum_x J_{\alpha,\beta}^\kappa(x),
  \end{equation}
  where we defined the $\kappa$-deformed generalized current operators as
\begin{equation}
  \label{Jabkappa}
i  \left[Q_\beta^\kappa,  q_\alpha^\kappa(x)\right]=J^\kappa_{\alpha,\beta}(x)-J^\kappa_{\alpha,\beta}(x+1).
\end{equation}

\item {\bf Bi-local operators.} We choose the generator in the form
   \begin{equation}
     \label{Xkappadef}
  X(\kappa)=
\sum_{x<y} q_\alpha^\kappa(x)q_\beta^\kappa(y).
\end{equation}
Repeating the computations of the previous Section the $\kappa$-derivatives can be written as
\begin{equation}
  \dperdk Q_\gamma^\kappa=  \sum_x K^\kappa_{\alpha,\beta,\gamma}(x),
\end{equation}
where the deformed $K$-operators are
\begin{equation}
  \label{Kabcdef}
  K^\kappa_{\alpha,\beta,\gamma}(x)=
  J^\kappa_{\alpha,\gamma}(x) q^\kappa_\beta(x)
    -q^\kappa_\alpha(x)J^\kappa_{\beta,\gamma}(x+1).
\end{equation}
This gives a recursion relation, which together with \eqref{Xkappadef} can be used to generate the
deformed charges. 
\end{enumerate}

In all of these cases there is a large gauge freedom inherent in specifying these generators. For example, the
charge density operators are not unique, as any re-definition of the form
\begin{equation}
  \label{gauge}
  q_\alpha(x)\quad\to\quad q_\alpha(x)+D(x+1)-D(x).
\end{equation}
leaves the charges invariant. Furthermore, instead of \eqref{Xkappadef} we could have a shifted
version as in  \eqref{Xl}. 
However, it can be shown that these choices do not affect the main conclusions: they always result
in adding extensive local operators to $X$, and do not affect the finite volume mean values of the perturbing operators.

It is an important property that the factorization \eqref{Jabfact} holds even for the deformed
operators given by \eqref{Kabcdef}. This can be seen by repeating the steps of the derivation of
\eqref{Jabfact}, and noticing that it does not use the locality properties of the operators, only
the global commutation of the charges and the local continuity equations.

\subsection{Action on Bethe states}

In this subsection we consider the deformation of the eigenstates dictated by the generating equation \eqref{defdef}.
This sets the lattice case apart from QFT, where it is not clear how the eigenstates are deformed.
Most local integrable spin chains can be solved by Bethe Ansatz and the eigenstates can be written down
explicitly. For simplicity we restrict ourselves to systems with a simple (non-nested) Bethe Ansatz; the extension to nested
cases is rather straightforward (c.f. \cite{beisert-long-range-2,sajat-longcurrents}). Furthermore, we do not specify the details of
the model; concrete formulas pertaining to the XXX and XXZ models can be found for example in
\cite{sajat-currents,sajat-longcurrents}.

The infinite volume (un-normalized) Bethe states with $N$ particles can be written as
\begin{equation}
  \label{psidef}
  \ket{\lan}=\sum_{x_1<x_2<\dots < x_N}
  \sum_{\sigma\in S_N}  \prod_{j>k} f(\lambda_{\sigma_j}-\lambda_{\sigma_k})
\prod_{j=1}^N e^{ip_{\sigma j}x_j}  \ket{x_1,\dots,x_N},
\end{equation}
where $\ket{x_1,\dots,x_N}$ are the basis states with the particles occupying positions $x_j$. We
have  $p_j=p(\lambda_j)$, where $p(\lambda)$ is the one-particle quasi-momentum
and $\lambda$ is the rapidity parameter. The summation  $\sigma\in S_N$ runs over all permutations of the rapidities. The function
$f(\lambda)$ describes the interaction between the particles with the scattering phase shift given by
\begin{equation}
  \label{Sdef}
  S(\lambda)=e^{i\delta(\lambda)}=\frac{f(\lambda)}{f(-\lambda)}.
  \end{equation}
In the above form the Bethe states are symmetric with respect to the exchange of rapidities.

The Bethe states are eigenvectors of the set of commuting charges with eigenvalues given by
\begin{equation}
  \label{Qmean}
  Q_\alpha\ket{\lan}=\Lambda_\alpha (\lan) \ket{\lan},\qquad
\Lambda_\alpha (\lan) =\sum_{j=1}^N h_\alpha(\lambda_j).
\end{equation}
The specific form of the one-particle eigenvalues is not relevant for us; explicit formulas can be found
 in \cite{sajat-currents,sajat-longcurrents}.

The total quasi-momentum of the state can be expressed as
\begin{equation}
\label{P0def}
  P=\sum_{j=1}^N p(\lambda_j).
\end{equation}

The deformation of the eigenstates can be understood simply from the generating equation \eqref{defdef}
\cite{beisert-long-range-2,sajat-longcurrents}. The Bethe states retain their functional form for
large separations of the particles, 
but the propagation factors and the relative phases change. Furthermore, the wave function acquires
additional correction 
terms for small separations. More precisely, for any finite $l$ contact terms appear for separations of $l$ sites in higher order terms $\sim
\kappa^c$ where $c\sim l$.

The boost-type deformations with charge $Q_\alpha$ generate a change of the one-particle momentum. It was shown in \cite{sajat-longcurrents}
that if this is the only ingredient in the deformation, then
\begin{equation}
  \label{pdef2}
  p^\kappa(\lambda)=p(\lambda)+\kappa h_\alpha(\lambda)
\end{equation}
holds to all orders in $\kappa$. The boost operators are one-particle irreducible, therefore they do not change the scattering matrix
of the particles. This property and the deformation of the momentum was used in \cite{sajat-longcurrents}  to derive the
finite volume mean values of the current operators; the results are summarized in Section \ref{sec:asymp} below.

The bi-local type deformations change the $S$-matrix of the spin chain \eqref{Sdef}, with the $\kappa$-derivative of
the scattering phase given by
\begin{equation}
  \label{deltadef}
  \dperdk \delta^\kappa(u,v)=h_\alpha(v)h_\beta(u)-h_\alpha(u)h_\beta(v).
\end{equation}
This equation can be understood intuitively by looking at the deformation of the states. The various
terms of the Bethe wave function get multiplied with different phases corresponding to the relative
positions of the particles, eventually resulting in the above anti-symmetric combination of the one-particle
charge eigenvalues. Note that the correction terms generally result in the $S$-matrix depending on the two rapidities separately instead of only on their difference, which in QFT corresponds to breaking Lorentz invariance.

The charge eigenvalues do not change under deformation, therefore the all-orders result is simply
\begin{equation}
  \label{deltadef2}
  \delta^\kappa(u,v)=\delta(u,v)+\kappa \left[h_\alpha(v)h_\beta(u)-h_\alpha(u)h_\beta(v)\right],
\end{equation}
and the deformation of $S$-matrix can be written is
\begin{align}
S(u,v)\mapsto S(u,v)e^{i\kappa(h_\alpha(v)h_\beta(u)-h_\alpha(u)h_\beta(v))},
\end{align}
which is precisely the lattice version of the QFT relation \eqref{eq:defSQFT}!

\section{Finite volume: asymptotic Bethe Ansatz}

\label{sec:asymp}

The long range deformations described in the previous Section are only defined in infinite volume. The reason for this is that
consistency requires the presence of correction terms with arbitrary long range which appear
at successively higher orders in $\kappa$. The problem in finite volume is that there is no general prescription
to define these long range terms once they wrap around the chain, which is the famous {\it wrapping problem}.
Nevertheless, the spectrum can be computed using the so-called {\it asymptotic Bethe Ansatz} up to exponentially
small corrections in the volume. The idea is to use the infinite volume quantities to construct the Bethe wave functions, and
set up the Bethe equations using this information. This procedure is justified up to some finite order in $\kappa$ such that the perturbing
operators still fit into the volume. Often we are only interested in the first order correction terms, and then the only requirement is that
the leading perturbing operator should fit into the volume \cite{sajat-longcurrents}.

We now collect the main asymptotic equations and derive the mean values of the perturbing
operators. Focusing on the bi-local case we also show that  the deformation of the scattering phase
\eqref{deltadef} is consistent with the factorisation \eqref{Jabfact}.

In finite volume the original Bethe equations are
\begin{equation}
  p_jL+\sum_{k\ne j}\delta(\lambda_j,\lambda_k)=2\pi I_j,
\end{equation}
where $I_j\in\egesz$ are the momentum quantum numbers.
The deformation is assumed to be continuous in $\kappa$, therefore the $I_j$ not be
changed. According to this the deformed asymptotic Bethe equations are
\begin{equation}
  p^\kappa_jL+\sum_{k\ne j}\delta^\kappa(\lambda_j,\lambda_k)=2\pi I_j,
  \label{eq:deformed_asymp_Bethe}
\end{equation}
where $p^\kappa$ and $\delta^\kappa$ are given by \eqref{pdef2} and \eqref{deltadef2}. Here we used
a notation reflecting that both quantities are changed at the same time. Whereas this is certainly possible by
combining the boost and bi-local types of generators, in the following we treat the two types of deformations
separately. We recall that the commutativity of the different types of deformations is discussed in
detail in \cite{beisert-long-range-2}. 

Within the asymptotic Bethe Ansatz the charge eigenvalues are computed as
\begin{equation}
  Q_\gamma^\kappa \ket{\lan}= \Lambda_\gamma^\kappa(L)
  \ket{\lan},\qquad
 \Lambda_\gamma^\kappa(L)    =\sum_{j=1}^N h_\gamma(\lambda_j).
\end{equation}
Therefore, the dependence of the charge eigenvalues on $\kappa$ comes entirely from the change of the
rapidities:
\begin{equation}
  \dperdk  \Lambda_\gamma^\kappa(L)=\sum_{j=1}^N h'_\gamma(\lambda_j) \frac{d\lambda_j}{d\kappa}.
\end{equation}
This idea was used in \cite{sajat-longcurrents} to derive the current mean values of the original
model, with the result
\begin{equation}
  \lanvev{J_{\alpha,\beta}}= {\bf h}'_\beta    G^{-1} {\bf h}_\alpha,
  \label{eq:meanvaluecurrent}
\end{equation}
where ${\bf h}_\alpha$ is a vector of length $N$ with components $h_\alpha(\lambda_j)$, ${\bf
  h}'_\beta$ is defined analogously, and $G$ is the Gaudin matrix of size $N\times N$ given by
\begin{equation}
  \label{gaudin2}
  G_{jk}=\frac{\partial}{\partial \lambda_k} \left(p_jL+\sum_{k\ne j}\delta(\lambda_j,\lambda_k)\right),\qquad
  j,k=1\dots N.
\end{equation}
Let us now also investigate the first order correction of  $\lanvev{Q_\gamma^\kappa}$ under the
bi-local transformation with $Q_\alpha$ and $Q_\beta$. Using perturbation theory this is given by
\begin{equation}
  \label{Qgchange}
  \dperdk \lanvev{Q_\gamma^\kappa}= L \lanvev{K_{\alpha,\beta,\gamma}(x)}.
\end{equation}
Taking the derivative of the asymptotic Bethe equations \eqref{eq:deformed_asymp_Bethe} with respect
to $\kappa$ yields 
\begin{equation}
  G_{jk} \frac{d\lambda_k}{d\kappa}+\sum_{k\ne j}\frac{d\delta(\lambda_j,\lambda_k)}{d\kappa}=0.
\end{equation}
From \eqref{deltadef}
\begin{equation}
  G_{jk} \frac{d\lambda_k}{d\kappa}=\sum_{k\ne j}\left[
  h_\alpha(\lambda_j)h_\beta(\lambda_k)-h_\alpha(\lambda_k)h_\beta(\lambda_j)\right].
 \end{equation}
The summation can be changed to include $j=k$, resulting in
\begin{equation}
  G_{jk} \frac{d\lambda_k}{d\kappa}=\left[
  h_\alpha(\lambda_j)\Lambda^\kappa_\beta-\Lambda^\kappa_\alpha h_\beta(\lambda_j)\right].
\end{equation}
which can be rewritten in vector form as
\begin{equation}
  \dperdk {\boldsymbol \lambda}=\Lambda^\kappa_\alpha  G^{-1} {\bf h}_\beta-
  \Lambda^\kappa_\beta  G^{-1} {\bf h}_\alpha.
\end{equation}
Substituting into  \eqref{Qgchange} results in
\begin{equation}
L\lanvev{K_{\alpha,\beta,\gamma}(x)}=
\Lambda^\kappa_\beta  \left(  {\bf h}'_\gamma    G^{-1} {\bf h}_\alpha\right)-
  \Lambda^\kappa_\alpha  \left( {\bf h}'_\gamma   G^{-1} {\bf h}_\beta\right).
\end{equation}
The eigenvalues divided by $L$ are nothing else but the mean values of the charge densities, while
the terms in parentheses are exactly the mean values of the currents given in
\eqref{eq:meanvaluecurrent}, so finally we obtain 
\begin{equation}
  \begin{split}
   &    \lanvev{K_{\alpha,\beta,\gamma}(x)}=\\
    &\hspace{1cm}=\lanvev{q_\beta(x)}\lanvev{J_{\alpha,\gamma}(x)}-
    \lanvev{q_\alpha(x)}\lanvev{J_{\beta,\gamma}(x)},
  \end{split}
  \end{equation}
which is exactly the factorisation condition \eqref{Jabfact} applied to the finite volume Bethe states.

Another approach is to start with the above factorisation equation (which was derived independently from Bethe
Ansatz), and then to compute the deformation rule \eqref{deltadef} by focusing on the two-particle case.
Alternatively, this result can also be considered as an independent derivation of the current mean values, starting
from the factorization \eqref{Jabfact} and the deformation \eqref{deltadef} of the scattering
phase. All these aspects can be 
summarised in the statement that the transformation rules are consistent with the algebra of the
charges and the currents. 

To conclude this Section let us return once more to the deformation rule \eqref{defdef}, which
  guarantees that the charge eigenvalues are invariant under the deformation {\it in infinite
    volume}. This is in contrast with the finite volume case, where the charge eigenvalues do change
  according to \eqref{Qgchange}. The difference is clearly due to the fact that the similarity
  transformation generated by  \eqref{defdef} is not  compatible with periodic boundary conditions.
  This phenomenon is completely analogous to what was
  observed in \cite{Cardy-TTbar1}, where it was explained that the $T\bar T$-deformation is eventually trivial
  {\it up to boundary terms}.

\section{Discussion}

\label{sec:discuss}

In this work we pointed out that the known bi-local type deformations of the integrable spin chains are
formally equivalent to the $T\bar T$-type deformations of integrable QFT. The key step in this
identification was the commutation
relation \eqref{XQg} of Theorem \ref{XQgthm}. We also found that the generalised current operators
play a central role, similar to the case of the
boost-type deformations treated in \cite{sajat-longcurrents}. Remarkably, the present
  framework yields the
  deformation of all conserved charges in a very straightforward way, which sets it
  apart from the known QFT computations which mostly focus on the Hamiltonian only.

A central result is the factorization of the mean values of the three index operators
$K_{\alpha,\beta,\gamma}(x)$, see eq. \eqref{Jabfact}.
The idea behind this is essentially the same as the original observations of Zamolodchikov regarding
the $T\bar T$-operator \cite{Zam-TTbar}.
Nevertheless, as far as we know, this form has not yet been written down for lattice systems.

Let us also mention, that for the XXZ spin chain the study of factorisation of correlation functions is of independent interest and
is by now well understood
\cite{bjmst-fact6,HGSIII,goehmann-kluemper-finite-chain-correlations,XXXfactorization}. In this
specific model the mean values of all local operators can be factorised, and our result \eqref{Jabfact} can be seen as a special case.
However, we also stress that our $K_{\alpha,\beta,\gamma}(x)$ are very specific operators that exhibit factorisation
for {\it every local integrable spin chain} and the factorisation also holds in the presence of the integrable long-range deformation.

Our observations open up many interesting future directions to explore.

\paragraph{Continuum limit}
It is well-known that the continuum limit and the low energy physics of the spin chains can in general be
described by a QFT. Therefore it is interesting to investigate the continuum limit of the deformed chains. While the bi-local deformation
for the spin chain is rather general and does not depend on specific details of the model, the continuum limit is more subtle and it does depend on
details regarding the states under consideration and the scaling to
the continuum limit.  All our perturbations are generally expected to
correspond to irrelevant operators, therefore their coupling constants
scale to zero under the Renormalization Group flow. Nevertheless it
might be possible to construct special scaling limits, such that the
resulting QFT's would correspond to the $T\bar{T}$ (or related)
deformations. This
connection
can potentially be 
exploited to learn about the short-distance physics and the ultraviolet completion of the $T\bar{T}$
deformations. 

\paragraph{Alternative formulations and relation to lattice gravity}
Even though the generating equation \eqref{defdef} provides all necessary details of the
deformation, we believe that the understanding of the long range spin chains is far from
complete. The present formulation is only tangentially related to the standard methods of lattice
integrability, such as local Lax operators, integrable vertex models, etc. Even though there are
certain special cases when a long range spin chain could indeed be related to the more conventional
models and methods
\cite{hubbard-and-long-range-1,didina-deform1,gromov-vieira-theta,yunfeng-didina-ivan}, the
generating equation \eqref{defdef}  opens up a big parameter space, which has not yet been
understood. It would be desirable to find an underlying integrable structure, which is rich enough
so that it could naturally accommodate all long  range deformations. Certain hints are provided by
QFT itself: in the Lagrangian formalism the  $T\bar T$ deformation turns out to be a pure boundary
effect \cite{Cardy-TTbar1}; furthermore, it can be understood by coupling the theory two a certain
2D gravity. It would be interesting to find  the natural lattice counterparts of these
phenomena. One possible lattice interpretation of the   $T\bar T$ deformation is found in the
upcoming work \cite{DJT}, where it is shown that the coupling to  2D gravity can be simulated by
allowing the inhomogeneity parameters of the spin chain to become dynamical variables. 

\paragraph{Other observables and non-equilibrium dynamics} The long range deformations lead to a
whole family of integrable spin chains whose properties have so far barely been investigated. It is
interesting to study their properties in more detail by computing physically interesting quantities
such as form factors, correlation functions and entanglement entropy. The non-local nature of the
long-range interacting spin chains is expected to result in a behaviour for certain
quantities which is qualitatively different from short range integrable spin chains. Given the
close connection between generalised hydrodynamics and long range deformation, it is also natural to
consider quantum quenches for these spin chains and see how the deformation changes the
non-equilibrium dynamics. 

\subsection*{Acknowledgments}

The authors are grateful to Marius de Leeuw, and D\'avid Horv\'ath for discussions. This work
was partially supported by the National Research Development and Innovation Office of Hungary under
grant  K-16 No.~119204 and also within the Quantum Technology National Excellence Program (Project
No. 2017-1.2.1-NKP-2017-00001). The work of G.T.  and B. P. was also partially supported by the
BME-Nanotechnology FIKP grant of ITM (BME FIKP-NAT), and B.P. also acknowledges support from the
J\'anos Bolyai Research Scholarship of the Hungarian
Academy of Sciences and the
\'UNKP-19-4 New National Excellence Program of the Ministry for Innovation and Technology.

\addcontentsline{toc}{section}{References}
\providecommand{\href}[2]{#2}\begingroup\raggedright\endgroup

\bibliographystyle{utphys}

\end{document}